\documentclass{article}
\usepackage{amssymb}
\usepackage{amsmath}
\usepackage{amsfonts}
\usepackage{amssymb}%
\usepackage{stmaryrd}
\usepackage{dsfont}
\usepackage{graphicx}

\usepackage{amsfonts}
\usepackage{amsmath}
\usepackage{amssymb}
\usepackage[english]{babel}
\usepackage[T1]{fontenc}
\usepackage{latexsym}
\usepackage{paralist}
\usepackage{rotating}
\usepackage{subfigure}
\usepackage{textcomp}
\usepackage{times}
\usepackage{url}
\usepackage{xcolor}
\usepackage{xspace}
\usepackage{graphicx}
\usepackage{epsfig}
\usepackage{paralist}
\usepackage{verbatim} 
\usepackage[section]{algorithm}
\usepackage{algorithmic}
\usepackage{tabularx}
\usepackage{multicol}

\newtheorem{proposition}{Proposition}

\newtheorem{proof}{Proof}
\newtheorem{definition}{Definition}

\begin{document}
%

\title{Low Cost Monitoring and Intruders Detection using Wireless Video Sensor Networks}

 \author{Jacques M. Bahi, Christophe Guyeux, Abdallah Makhoul\\ 
\footnotesize%
 	University of Franche-Comt\'e (LIFC)\\
 	\footnotesize%
 	 Rue Engel-Gros, BP 527, 90016 Belfort Cedex, France\\
 	 \footnotesize%
   E-mails: \{Jacques.Bahi, Mourad.Hakem, Abdallah.Makhoul\}@univ-fcomte.fr \\
   ~\\
      Congduc Pham\\
   \footnotesize%
 	University of Pau (LIUPPA)\\ 
 	\footnotesize%
 	Avenue de l'Universit\'e, BP 1155, 64013 Pau France\\
   \footnotesize%
   E-mail: \{congduc.pham\}@univ-pau.fr 
   }

\maketitle

\begin{abstract}        

There is a growing interest in the use of video sensor networks in surveillance applications in order to detect intruders with low cost. The essential concern of such networks is whether or not a specified target can pass or intrude the monitored region without being detected. This concern forms a serious challenge to wireless video sensor networks of weak computation and battery power.

In this paper, our aim is to prolong the whole network lifetime while fulfilling the surveillance application needs. We present a novel scheduling algorithm where only a subset of video nodes contribute significantly to detect intruders and prevent malicious attacker to predict the behavior of the network prior to intrusion. Our approach is chaos-based, where every node based on its last detection, a hash value and some pseudo-random numbers easily computes a decision function to go to sleep or active mode. We validate the efficiency of our approach through theoretical analysis and demonstrate the benefits of our scheduling algorithm by simulations. Results show that in addition of being able to increase the whole network lifetime and to present comparable results against random attacks (low stealth time), our scheme is also able to withstand malicious attacks due to its fully unpredictable behavior.

\end{abstract}

\section{Introduction}

Instead of using traditional vision systems built essentially
from fixed video cameras, it is possible to deploy autonomous
and small wireless video sensor nodes (WVSN)~\cite{51} to achieve video surveillance of a given area of interest. Doing so lead to a
much higher level of flexibility, therefore extending the range
of surveillance applications that could be considered. More
interestingly, this scenario can support dynamic deployment
scenario even in so-called object and obstacle-rich environments
or hard-to-access areas. Such wireless video sensor
nodes can in addition be thrown in mass to constitute a large
scale surveillance infrastructure. In these scenarios, hundreds
or thousands of video nodes of low capacity (resolution,
processing, and storage) of a same or similar type can be
deployed in an area of interest.

Surveillance applications have very specific needs due to
their inherently critical nature associated to security~\cite{68,69,70}. 
The basic objective of video surveillance systems is to allow detection 
and/or identification of intruders. Therefore, in that context, the main goal of a video sensor 
network is to ensure the coverage of the whole area of interest at any time $t$. 
Another issue of prime importance is related to energy considerations since the scarcity of energy does have a direct impact on coverage, as it is not possible to have all the video nodes in activity at the same time. 
Therefore, a common approach is to define a subset of the deployed nodes to be active while the other nodes can sleep. 
There are already some techniques that schedule video nodes to work alternatively while maintaining the complete coverage~\cite{72,73,74}. 
The main idea in these techniques is to turn off a redundant node. 
Here redundancy means that the covered area by a node is completely covered by its neighbors too. 
However, these techniques usually depend on location or directional information, which is costly in energy and complexity. 
Usually it is very difficult to determine the redundant nodes without the location information. 
Fortunately, not all applications need a complete coverage at anytime, and in most surveillance applications for intrusion detection, most sensor nodes can move to a so-called ``idle mode'' in the absence of intrusions. When an intruder is detected by a node all the network will be alerted. In that context, it is critical to provide an effective scheme for turning off video nodes without degrading the surveillance quality.

In this paper, we present a solution to the joint scheduling problem in surveillance applications using video sensor nodes. 
We provide a chaotic sleeping scheme and conduct a theoretical and simulation analysis of both performances and security. 
Until now, only random approaches have been extensively studied in the literature to turn off video nodes without degrading the surveillance quality. 
Even if such methods present good scores in detecting random intrusions while preserving the lifetime of the network, they do not encompass the situation of a malicious attacker. 
That is to say, the intruder is not supposed to know something about the surveillance scheme, he cannot observe the WVSN for a while, or he is not authorized to deduce anything from his possible knowledge. 
In this paper, we intend to tackle with situations where the attacker is not supposed passive: he is smart and does not necessarily choose a random way to achieve his intrusion. 
In addition of preserving the network lifetime and being able to face random attacks, we show that our scheme is also capable to withstand attacks of a malicious adversary due to its unpredictable behavior.

The rest of the paper is organized as follows.
In Section \ref{section:related work}, related works related to surveillance applications with WVSN are presented.
Smart threats and malicious attackers are introduced in Section \ref{section:Smart Threat}.
Basic recalls and terminologies on the fields of the mathematical theory of chaos and chaotic iterations are given in Section \ref{section:Basic recalls}, and the link unifying them is explained too.
The surveillance scheme based on the chaos theory is detailed in Section \ref{chaos-sch}. 
We show in Section \ref{section:theoretical study} that our proposed scheme can be used against malicious attacks. Simulation results in Section \ref{section: simulation results} compare our scheme to the classical random schedule in terms of intruder's stealth time, network lifetime and energy repartition. The paper ends by a conclusion section, where our contribution is summed up and planned future work is detailed.

\section{Related Works}
\label{section:related work}

In video sensor networks, minimizing energy consumption and prolonging the system lifetime are major design objectives. 
Due to the significant energy-saving when a node is sleeping, a frequently used mechanism is to schedule the sensor nodes such that redundant nodes go to sleep as often and for as long as possible. 
By selecting only a subset of nodes to be active and keeping the remaining nodes in a sleep state, the energy consumption of the network is reduced, thereby extending the operational lifetime of the sensor network. 

In this context, the coverage problem for wireless video sensor networks can be categorized as:
\begin{itemize}
\item {\it Known-Targets Coverage Problem,} which seeks to determine a subset of connected video nodes that covers a given set of target-locations scattered in a 2D plane.
\item {\it Region-Coverage Problem,} which aims to find a subset of connected video nodes that ensures the coverage of the entire region of deployment in a 2D plane.
\end{itemize}

Most of the previous works have considered the known-targets coverage problem~\cite{76,77,78,80}. 
The objective is to ensure at all time the coverage of some targets with known locations that are deployed in a two-dimensional plane.
For example, the authors in~\cite{80} organize sensor nodes into mutually exclusive subsets that are activated successively, where the size of each subset is restricted and not all of the targets need to be covered by the sensors in one subset.
In~\cite{77}, a directional sensor model is proposed, where a sensor is allowed to work in several directions.
The idea behind this is to find a minimal set of directions that can cover the maximum number of targets.
It is different from the approach described in~\cite{76} that aims to find a group of non-disjoint cover sets, each set covering all the targets to maximize the network lifetime.

Regarding the Region-Coverage Problem in which this study takes place, existing works focus on finding an efficient deployment pattern so that the average overlapping area of each sensor is bounded. 
The authors in~\cite{81} analyze new deployment strategies for satisfying some given coverage probability requirements with directional sensing models. 
A model of directed communications is introduced to ensure and repair the network connectivity.
Based on a rotatable directional sensing model, the authors in~\cite{75} present a method to deterministically estimate the amount of directional nodes for a given coverage rate. 
A sensing connected sub-graph accompanied with a convex hull method is introduced to model a directional sensor network into several parts in a distributed manner. 
With adjustable sensing directions, the coverage algorithm tries to minimize the overlapping sensing area of directional sensors only with local topology information.
Lastly, in~\cite{73}, the authors present a distributed algorithm that ensures both coverage of the deployment area and network connectivity, by providing multiple cover sets to manage Field of View redundancies and reduce objects disambiguation.

All the above algorithms depend on the geographical location information (position and direction) of video nodes. 
These algorithms aim to provide a complete-coverage network so that any point in the target area would be covered by at least one video node. 
However, this strategy is not as energy-efficient as what we expect because of the following two reasons.
Firstly, the energy cost and system complexity involved in obtaining geometric information may compromise the effect of those algorithms. 
Secondly, video nodes located at the edge of the area of interest must be always in an active state as long as the region is required to be completely covered. 
These video nodes will die after some time and their  coverage area will be left without surveillance. 
Thus, the network coverage area will shrink gradually from outside to inside. 
This condition is unacceptable in video surveillance applications and intrusion detection, because the major goal here is to detect intruders as they cross a border or as they penetrate a protected area.

One direction to solve these problems is to schedule a node to sleep following a probabilistic approach. 
Each node remains awake with a given probability so that the coverage of the area can be guaranteed.
However the probability can be modeled by an observer, who can take benefits from his observations to predict the dynamic of the network.
This is obviously a security flaw.
These considerations lead us to the introduction of smart threats given in the next section.

\section{Smart Threats}
\label{section:Smart Threat}

\subsection{Introduction}

Let us suppose that an adversary tries to reach a location $X$ into the area without being detected. 
We consider that this situation leads to two categories of attacks against WVSN surveillance. 

On the one hand, the attacker only knows that the area is under surveillance. 
He tries to take its chance, for example by following the shortest way or by trying a random path. 
In this first category of attack that we call ``blind elementary attacks'', the intruder does not know how the surveillance is achieved as he does not observe the WVSN.

On the other hand, in the second category of attacks, called ``malicious attacks'' in this paper, the intruder is supposed to be intelligent. 
He can try to take benefits from his observations to understand the behavior of the WVSN. 
After having recorded the dynamic of the WVSN for a given time, the malicious intruder can try to determine when video nodes are turned on. 
This prediction can help the intruder to find a way to reach $X$ without being detected.

In our opinion, the most reasonable way to evaluate the consequences of a malicious attack is to suppose that the intruder has access to the surveillance scheme. 
With this supposition, our security model encompasses the case where an attacker can have a physical access to a given node, thus determining the embedded mechanism used for video surveillance.
In this Kerckhoffs-based principle, the attacker knows all but the initial parameters of the nodes.
Moreover, he can observe the WVSN for a while.
To achieve his intrusion, he can use all of the acquired knowledge -- the sole difficulty is his lack of a secret parameter (the secret key) used to initialize the surveillance process.

The context of blind elementary attacks is well-known and understood: it has been studied a lot in the last decade, and various solutions have yet been proposed (Section \ref{section:related work}). 
On the contrary, to the best of our knowledge, the case of an intelligent intruder (smart threat) has not yet really been treated. 
In this paper, we intend to propose a scheme able to withstand attacks encompassing these malicious intrusions, and thus to offer a first solution to the problem raised by the smart threats existence hypothesis.

Technically speaking, the proposed approach offers several benefits. 
Firstly, the node scheduling algorithm does not need location information. 
Therefore, the energy consumption is reduced because there is no need to locate the node itself and its neighbors. 
Secondly, we will show that it performs as well as a random scheduling, in terms of lifetime and intrusion detection against blind elementary attacks (see Section \ref{section: simulation results}). 
Lastly, due to its chaotic properties, its coverage is unpredictable, and thus a malicious adversary has no solution to attack the network (Section \ref{section:theoretical study}).

\subsection{Classification of Malicious Attacks}

When a malicious adversary attacks a WVSN, he can concentrate his efforts 
either on the global network or on some specific nodes. Depending on the
considered situation, he can perform either an active attack, modifying 
the network architecture or a node, or a passive attack based only on 
observations. He can have access to several WVSN using the
same algorithm. Furthermore, he can build its own network to make some 
experiments. His objective is to find the secret key used in the targeted 
network: with this knowledge, the attacker will be able to predict the 
behavior of the video sensor nodes.

Active attacks have been already investigated several times in the literature.
These studies encompass the cases where nodes can be added, moved, modified, or 
removed, where communications between nodes can be observed or modified, and
where the global architecture of the network is attacked. However, some WVSN
are such that any modification of the network is signaled, leading to the 
impossibility of such active attacks. On the contrary, passive observations 
and deductions of a malicious attacker are always possible. To the best
of our knowledge, these threats have not yet been investigated.

The passive malicious attacks can be classified as follows.

\begin{itemize}
 \item In the \textbf{Target Only Attack (TOA)}, the adversary can only 
observe targeted networks.
 \item In the \textbf{Constant Key Attack (CKA)}, the adversary has access
to several WVSN using the same secret key. The areas under surveillance and
the network architecture change from one WVSN to another, but the attacker
knows that all these networks use the same algorithm with the same secret
key.
 \item In the \textbf{Known Original Attack (KOA)}, the attacker had previously 
accessed to the WVSN and its area. He had the opportunity to test various 
keys in a previous access. He hopes that this knowledge will help him to 
determine a way to realize his intrusion when the WVSN is really launched.
 \item In the \textbf{Chosen Key Attack (CKA)}, the adversary has access
to an exact copy of the network and area under surveillance than the one he
want to attack. He has realized for instance a miniature model or a computer
simulator having exactly the same behavior than the targeted network 
and its area. He can thus try several secret key, and if he achieves to 
reproduce exactly the same behavior for the network, then he can reasonably
suppose that the true secret key has been discovered.
 \item Finally, in the \textbf{Estimated Original Attack (EOA)}, the attacker has
only an estimation, an approximation of the network and its area.  
\end{itemize}

In each of these categories, the sole objective of the attacker is to 
obtain the value of the secret key. With this knowledge, he will able to
determine the WVSN behavior, finding by doing so a way to achieve his
intrusion.

\subsection{Security Levels in CKA}

We now take place in the Chosen Key Attack problem. Let $k_0$ be the secret
key used to initiate the video-surveillance. Denote by 
$Y_k$ the probabilistic model that the attacker can build with his observations,
and by $\mathds{K}$ the set of all possible keys.

\begin{definition}[Insecurity]
 The WVSN is insecure against the Target Only Attack if and only if 
$\exists k_1 \in \mathds{K}, p(Y_{k_1}) = p(Y_{k_0})$ and
$\forall k_2 \in \mathds{K}, p(Y_{k_2}) \neq p(Y_{k_0})$ 
\end{definition}

On the contrary,

\begin{definition}[Security]
 The WVSN is secure against the Target Only Attack if and only if 
$\forall k \in \mathds{K}, p(Y_{k}) = p(Y_{k_0})$ 
\end{definition}

In that situation, it is easy to prove that the mutual information 
$\mathcal{I}(k_0,Y_{k_0})$ is equal to 0, which is often refered 
as \emph{perfect secrecy}.

\section{Basic Recalls}
\label{section:Basic recalls}

In the sequel $S^{n}$ denotes the $n-$th term of a sequence $S$ and $V_{i}$ is the $i-$th component of a vector $V$. $f^{k}=f\circ ...\circ f$ denotes the $k-$th composition of a function $f$. Finally, the following notation is used: $\llbracket1;N\rrbracket=\{1,2,\hdots,N\}$.

\subsection{Devaney's Chaotic Dynamical Systems}
\label{subsection:Devaney}

Consider a topological space $(\mathcal{X},\tau)$ and a continuous function $f : \mathcal{X} \rightarrow \mathcal{X}$. 

\begin{definition}
$f$ is said to be \emph{topologically transitive} if, for any pair of open sets $U,V \subset \mathcal{X}$, there exists $k>0$ such that $f^k(U) \cap V \neq \varnothing$.
\end{definition}

\begin{definition}
An element (a point) $x$ is a \emph{periodic element} (point) for $f$ of period $n\in \mathds{N}^*,$ if $f^{n}(x)=x$.
\end{definition}

\begin{definition}
$f$ is said to be \emph{regular} on $(\mathcal{X}, \tau)$ if the set of periodic points for $f$ is dense in $\mathcal{X}$: for any point $x$ in $\mathcal{X}$, any neighborhood of $x$ contains at least one periodic point (without necessarily the same period).
\end{definition}

\begin{definition}
$f$ is said to be \emph{chaotic} on $(\mathcal{X},\tau)$ if $f$ is regular and topologically transitive.
\end{definition}

The chaos property is strongly linked to the notion of ``sensitivity'', defined on a metric space $(\mathcal{X},d)$ by:

\begin{definition}
\label{sensitivity} $f$ has \emph{sensitive dependence on initial conditions}
if there exists $\delta >0$ such that, for any $x\in \mathcal{X}$ and any neighborhood $V$ of $x$, there exists $y\in V$ and $n > 0$ such that $d\left(f^{n}(x), f^{n}(y)\right) >\delta $.
$\delta$ is called the \emph{constant of sensitivity} of $f$.
\end{definition}

Indeed, Banks \emph{et al.} have proven in~\cite{Banks92} that when $f$ is chaotic and $(\mathcal{X}, d)$ is a metric space, then $f$ has the property of sensitive dependence on initial conditions (this property was formerly an element of the definition of chaos). To sum up, quoting Devaney in~\cite{Devaney}, a chaotic dynamical system ``is unpredictable because of the sensitive dependence on initial conditions. It cannot be broken down or simplified into two subsystems which do not interact because of topological transitivity. And in the midst of this random behavior, we nevertheless have an element of regularity''. Fundamentally different behaviors are consequently possible and occur in an unpredictable way.



\subsection{Chaotic Iterations}
\label{sec:chaotic iterations}

Let us consider a \emph{system} of a finite number $\mathsf{N} \in \mathds{N}^*$ of elements (or \emph{cells}), so that each cell has a Boolean \emph{state}. 
A sequence of length $\mathsf{N}$ of Boolean states of the cells corresponds to a particular \emph{state of the system}. 
A sequence which elements are subsets of $\llbracket 1;\mathsf{N} \rrbracket $ is called a \emph{strategy}. 
The set of all strategies is denoted by $\mathbb{S}.$

\begin{definition}
The set $\mathds{B}$ denoting $\{0,1\}$, let $f:\mathds{B}^{\mathsf{N}}\longrightarrow \mathds{B}^{\mathsf{N}}$ be a function and $S\in \mathbb{S}$ be a strategy. 
The so-called \emph{chaotic iterations} (CIs) are defined by~\cite{Robert1986} $x^0\in \mathds{B}^{\mathsf{N}}$ and
\begin{equation}
\forall n\in \mathds{N}^{\ast }, \forall i\in \llbracket1;\mathsf{N}\rrbracket ,x_i^n=\left\{
\begin{array}{ll}
x_i^{n-1} & \text{ if } i \notin S^n \\
\left(f(x^{n-1})\right)_{S^n} & \text{ if } i \in S^n.\end{array}\right.
\end{equation}
\end{definition}

In other words, at the $n-$th iteration, only the $S^{n}-$th cell is \textquotedblleft iterated\textquotedblright .

Note that in a more general formulation, $S^n$ can be a subset of components and $f(x^{n-1})_{S^{n}}$ can be replaced by $f(x^{k})_{S^{n}}$, where $k<n$, describing for example, delays transmission.
For the general definition of such chaotic iterations, see, \emph{e.g.}, \cite{Robert1986}.

The term ``chaotic'', in the name of these iterations, has \emph{a priori} no link with the mathematical theory of chaos recalled previously.
However, we have proven in \cite{guyeux09} that in a relevant metric space $(\mathcal{X},d)$, the vectorial negation $f_{0}(x_{1},\hdots,x_{\mathsf{N}})=(\overline{x_{1}},\hdots, \overline{x_{\mathsf{N}}})$ satisfies the three conditions for Devaney's chaos. 
This result is recalled in the next section.

\subsection{Chaotic Iterations and Devaney's Chaos}
\label{sec:topological}

Denote by $\Delta $ the \emph{discrete Boolean metric}, $\Delta(x,y)=0\Leftrightarrow x=y.$ Given a function $f:  \mathds{B}^{\mathsf{N}}\longrightarrow \mathds{B}^{\mathsf{N}}$, define the function $F_{f}:$ $\llbracket1;\mathsf{N}\rrbracket\times \mathds{B}^{\mathsf{N}}\longrightarrow \mathds{B}^{\mathsf{N}}$ such that
\begin{equation}
F_{f}(k,E)=\left( E_{j}.\Delta (k,j)+f(E)_{k}.\overline{\Delta (k,j)}\right)_{j\in \llbracket1;\mathsf{N}\rrbracket},
\end{equation}

\noindent where + and . are the Boolean addition and product operations.
The \emph{shift} function is defined by $\sigma :(S^{n})_{n\in \mathds{N}}\in \mathbb{S}\rightarrow (S^{n+1})_{n\in \mathds{N}}\in \mathbb{S}$ and the \emph{initial function} $i$ is the map which associates to a sequence, its first term: $i:(S^{n})_{n\in \mathds{N}}\in \mathbb{S}\rightarrow S^{0}\in \llbracket1;\mathsf{N}\rrbracket$.

Consider the phase space: $\mathcal{X}=\llbracket1;\mathsf{N}\rrbracket^{\mathds{N}}\times \mathds{B}^{\mathsf{N}}$ and the map
\begin{equation}
G_{f}\left( S,E\right) =\left( \sigma (S),F_{f}(i(S),E)\right).
\end{equation}

The chaotic iterations can be described by the following iterations
\begin{equation}
\left\{
\begin{array}{l}
X^{0}\in \mathcal{X} \\
X^{k+1}=G_{f}(X^{k}).
\end{array}
\right.
\end{equation}

Let us define a new distance between two points $(S,E),(\check{S},\check{E} )\in \mathcal{X}$ by
\begin{equation}
d((S,E);(\check{S},\check{E}))=d_{e}(E,\check{E})+d_{s}(S,\check{S}),
\end{equation}

where
\begin{itemize}
\item $\displaystyle{d_{e}(E,\check{E})}=\displaystyle{\sum_{k=1}^{\mathsf{N}}\Delta (E_{k},\check{E}_{k})} \in \llbracket 0 ; \mathsf{N} \rrbracket$,
\item $\displaystyle{d_{s}(S,\check{S})}=\displaystyle{\dfrac{9}{\mathsf{N}}\sum_{k=1}^{\infty }\dfrac{|S^{k}-\check{S}^{k}|}{10^{k}}} \in [0 ; 1].$
\end{itemize}

This new distance has been introduced in \cite{guyeux10,guyeux09} to satisfy the following requirements. When the number of different cells between two systems is increasing, then their distance should increase too. In addition, if two systems present the same cells and their respective strategies start with the same terms, then the distance between these two points must be small because the evolution of the two systems will be the same for a while. The distance presented above follows these recommendations. Indeed, if the floor value $\lfloor d(X,Y)\rfloor $ is equal to $n$, then the systems $E, \check{E}$ differ in $n$ cells. In addition, $d(X,Y) - \lfloor d(X,Y) \rfloor $ is a measure of the differences between strategies $S$ and $\check{S}$. More precisely, this floating part is less than $10^{-k}$ if and only if the first $k$ terms of the two strategies are equal. Moreover, if the $k-$th digit is nonzero, then the $k-$th terms of the two strategies are different. 

\medskip

It is proven in \cite{guyeux10,guyeux09} by using the sequential continuity that,

\begin{proposition}
\label{continuite} 
$\forall \mathsf{N} \in \mathds{N}^*, \forall f:\mathds{B}^\mathsf{N} \to \mathds{B}^\mathsf{N}, G_{f}$ is a continuous function on $(\mathcal{X},d)$.
\end{proposition}

It is then checked in \cite{guyeux10,guyeux09} that in the metric space $(\mathcal{X},d)$, the vectorial negation $f_{0}(x_{1},\hdots,x_{\mathsf{N}})=(\overline{x_{1}},\hdots, \overline{x_{\mathsf{N}}})$ satisfies the three conditions for Devaney's chaos: regularity, transitivity, and sensitivity. This has led to the following result.

\begin{proposition}
CIs are chaotic on $(\mathcal{X},d)$ as it is defined by Devaney.
\end{proposition}

These chaotic iterations have been used to define in \cite{guyeux09} an hash function and a pseudo-random number generator (PRNG) able to pass the stringent TestU01 battery of tests in \cite{bgw10:ip}.

\section{Chaos-Based Scheduling}
\label{chaos-sch}

\subsection{The General Algorithm}

\subsubsection{Network Capabilities}

The WVSN is supposed to be constituted by $2^{\mathsf{N}}$ nodes $V_i, i\in \llbracket 0, 2^{\mathsf{N}}-1 \rrbracket$. 
Each $V_i$ is able to wake up on a specific signal, to survey a given area (and to detect intrusions), to send a wake up signal to another node $V_j$, and to go to sleep when it is required. Furthermore, it is supposed that $V_i$ embeds:

\begin{itemize}
\item The mechanisms required by the intrusion detection: a sensing function $c_i(t)$, such as a camera, which returns some digital data at each listening time, and a decision function $d_i(c)$ which returns if an intrusion is detected in this sensing values ($c_i(t)$) or not.
\item An internal clock having the time $T_i = r_i T_0$ as a reference.
\item A vector of $\mathsf{N}$ binary digits, called \emph{the state of the system} $V_i$, and the capability to swap each bit of this vector ($0 \leftrightarrow 1$).
\item An integer $e_i$, called \emph{listening time}, initialized to 0.
\end{itemize}

In other words, each node $V_i$ can achieve CIs. 
Thus, each node can compute, easily and by using a few resources, a hash value and some pseudo-random numbers as it is recalled in Section \ref{sec:chaotic iterations}. 
We will denote by $g_i$ the seed of the PRNG used in node $V_i$, which is equal to a secret parameter $p_i$ at time $t=0$.
This secret parameter with $\mathsf{N}$ bits has been generated
by a cryptographically secure PRNG, and thus it is uniformly
distributed into $\llbracket 0; 2^\mathsf{N}-1 \rrbracket$.
The state $V_i$ is initialized to the binary decomposition
of $g_i$.

\subsubsection{Deploying the Network}

The deployment of video sensor nodes in the physical environment is the first operation (step)
in the network lifecycle. It may take several forms. Sensor nodes may be randomly deployed dropping them from a plane, and placed one by one by a
human or a robot. Deployment may be a one time activity or a continuous process. 
These methods have been extensively studied in the literature.
In our method, the sole requirement to satisfy is to guarantee the uniform repartition into the region of interest.

\subsubsection{Initialization of the WVSN}

At time $t=0$, a subset $\mathcal{I} \subset \llbracket 0, 2^\mathsf{N}-1 \rrbracket$ of nodes are woken up and $\forall i \in \mathcal{I}, e_i^{t_0} = T_i$.

\subsubsection{Surveillance}

The principle of surveillance applications is defined as follows. 
At each time $t_j = j \times T_0, j = 1, 2, \hdots $:

\begin{enumerate}
\item If a sleeping node $V_i$ has received $n_i^{t_j-1} \geqslant 1$ wake up orders during the time interval $[t_{j}-1,t_j]$, then it goes into active mode and sets its listening time $e_i^{t_j}$ to $n_i^{t_j-1} T_i$.
\item If an active node $V_i$ has received $n_i^{t_j-1} \geqslant 1$ orders to wake up during the time interval $[t_{j-1},t_j]$, then it increments its listening time: $e_i^{t_j} = e_i^{t_j-1}+n_i^{t_j-1} T_i$.
\item For each node $V_i$ having a listening time $e_i^{t_j} \neq 0$:
\begin{itemize}
\item $V_i$ ensures the surveillance of its area during $T_0$,
\item If, during this time interval, an intrusion is detected, then the WVSN is under alert.
\item If $t_j$ is the first listening time of $V_i$ after having activated, then:
\begin{itemize}
\item The hash value $h_i^{t_j}$ of the sensed value $c_i(t_j)$ is computed (cf. Section~\ref{sec:chaotic iterations}).
\item The seed $g_i$ of the PRNG of $V_i$ is set to $h_i^{t_j}+t_j$, where $+$ is the concatenation of the digits of $h_i^{t_j}$ and $t_j$ (thus even if $h_i^{t_j}=h_i^{t_k}, k<j$, we have $g_i^{t_j} \neq g_i^{t_k}$).
\item The $\mathsf{N}$ bits of the state of the system $V_i$ are set to $E_i^{t_j}$, where $E_i^{t_j}$ is the binary decomposition of $i$ shown as a binary vector of length $\mathsf{N}$.
\end{itemize}
\end{itemize}
\item $\mathsf{N}$ bits are computed with the PRNG of $V_i$. These bits define an integer $S_i^{t_j} \in \llbracket 0, 2^\mathsf{N}-1 \rrbracket$. Then the bit of $E_i^{t_j}$ in position $S_i^{t_j}$ is switched, which leads to a new state $E_i^{t_{j+1}}$. By doing so, CIs are realized.
\item Each active node $V_i$ decreases its listening time: $e_i^{t_j} = e_i^{t_j} - 1$.
\item For each active node having its listening time $e_i^{t_j} = 0$:
\begin{itemize}
\item $V_i$ sends the wake up order to node $V_k$, where $k \in \llbracket 0, 2^\mathsf{N}-1 \rrbracket$ is the integer whose binary decomposition is the last state of the system $V_i$ ($E_i^{t_{j+1}}$). 
\item $V_i$ goes to sleep.
\end{itemize}
\end{enumerate}

\section{Theoretical Study}
\label{section:theoretical study}

\subsection{Scheduling as Chaotic Iterations}

\label{subsec:Scheduling as Chaotic Iterations}
The scheduling scheme presented above can be described as CIs.
The global state $E^t$ of the whole system is constituted by the reunion of each internal state $E_i^t$ of each node $i$.
This is an element of $\mathds{B}^{\mathsf{N}\times 2^\mathsf{N}}$.
The strategy at time $t$ is the subset of $\llbracket 0; \mathsf{N}\times 2^\mathsf{N} \rrbracket$ constituted by all of the strategies that are computed into the awaken nodes at time $t$.
More precisely, if the node $V_k$ has computed the strategy $S_k^t$ at time $t$, then the global strategy $S^t$ will contain the value $S_k^t+k\times \mathsf{N}$.
Lastly, the iteration function is the vectorial negation defined $:\mathds{B}^{\mathsf{N}\times 2^\mathsf{N}} \rightarrow \mathds{B}^{\mathsf{N}\times 2^\mathsf{N}}$.
A subsequence $E^{m^t}$ is extracted from $E^t$, which determines the changes that occur in the network: nodes whose binary id is into $E^{m^t}$ are nodes that achieve the surveillance at the considered time.
Let us remark that $S^k$ and $m^k$ depend both on the outside world, due to the fact that $S_i^t$ are regularly seeded with the digest of some sensed values.

\subsection{Complexity}

Even if the hash function and the PRNG taken from \cite{guyeux09} and \cite{bgw10:ip} respectively can be replaced by any cryptographically secure hash function and PRNG, we do not recommend their substitution.
Indeed, all of the operations used by our scheme can be achieved by CIs.
Each iteration of CIs is only constituted by the negation of a few binary digits.
Obviously, such an operation is fast and does not consume a lot of energy.
By doing so, we thus obtain an efficient video surveillance scheduling scheme compliant with WVSN requirements.
Section \ref{section: simulation results} will detail more quantitatively this fact.

\subsection{Coverage}

The coverage of the whole area is guaranteed due to the following reasons.

Firstly, the scheduling process corresponds to CIs.
These iterations are chaotic according to Devaney, thus they are transitive.
This transitivity property is the formulation of an uniform distribution in terms of topology.
It claims that the system will never stop to visit any sub-region of the whole area, regardless of how tiny the region is.

Secondly, as the choice of the nodes to wake up at each time are done by using CIs, this selection corresponds to the returned value of our PRNG proposed in \cite{bgw10:ip}.
This ``CI(X,Y)-generator'' takes two PRNGs X,Y as input sequences, realizes CIs with X as strategy, the vectorial negation as update function, and selects the states to publish as outputs by using the second PRNG Y.
By such a combination, we improve the statistical properties of the input PRNG used as strategy, and we add chaotic properties.
The scheduling process corresponds to the CI(X,Y)-generator, with X=$m$ and Y=$S$.
As Y is statistically perfect (Y is CI(ISAAC,ISAAC), which can pass the whole NIST, DieHARD, and TestU01 batteries of tests), the uniform repartition of the states is then guaranteed.

Lastly, experiments of Section \ref{section: simulation results} will show that this intended uniform coverage is well obtained in practice.

\subsection{Security Study}

\subsubsection{Qualitative Approach}

Let us suppose that Oscar, an intruder, knows that the scheduling process is based on CIs, i.e. he knows the whole algorithm, except the seeds that have been used to initiate the PRNGs of each node.
By doing so, we respect the Kerckhoffs' principle: the adversary has all except the secret key.
Oscar's desire is to reach a particular location $X$ of the area without being detected.
To achieve his goal, he can choose two strategies. 
On the one hand, he can try a blind elementary attack, either by following a random way from its position to $X$, of by choosing the shortest path.
The next subsection and the experiments will show that such an attack cannot work.
On the other hand, Oscar can try to take benefits both from his knowledge and his observations.
However, if he can determine the nodes that are awaken at time $t$, he cannot predict the awaken nodes at time $t+1, t+2, ...$
To do so, he should be able to obtain $S^{t+1}, S^{t+2},...$, which are computed from the digests of some values that will be sensed in the future. 
As our hash function satisfy the avalanche effect, due to its chaotic properties, any error on the sensed value lead to a completely different digest.

As Oscar cannot determine the sensed values of each node, at each time and with an infinite precision, he does not have the knowledge of the current state of the global system.
He has only access to an approximation of this state. 
As the global scheduling process is chaotic, this error on the initial condition is magnified at each iteration, leading to the impossibility for Oscar to predict the scheduling process.
This qualitative approach for security will be formalized in the second section below.

\subsubsection{Chaotic Properties}

We now investigate the topological properties presented by the proposed 
video-surveillance scheme.
First  of  all, let us  recall  two  fundamental  definitions from  the
mathematical theory of chaos.

\begin{definition} 
A function $f$ is said  to be {\bf expansive} if $\exists \varepsilon>0$,
$\forall    x   \neq   y$,    $\exists   n    \in   \mathds{N}$    such   that
$d\left(f^n(x),f^n(y)\right) \geq \varepsilon$.
\end{definition}

\begin{definition} 
A  discrete  dynamical  system   is  said  to  be  {\bf  topologically
mixing} if and only if, for  any pair of disjoint open sets $U$,$V
\neq \emptyset$,  $n_0 \in \mathds{N}$  can be found  so that $\forall  n \geq
n_0$, $f^n(U) \cap V \neq \emptyset$.
\end{definition}

As  proven in  \cite{gfb10:ip}, chaotic  iterations are  expansive and
topologically  mixing  when  $f$  is  the  vectorial  negation  $f_0$.
Consequently,  these  properties are  inherited  by the WVSN
presented  previously, which induce a greater
unpredictability.  Any difference on  the initial parameter of  the WVSN
is  in particular  magnified up to  be equal to  the expansivity
constant. 

Now,  what are the  consequences for  a wireless sensor network to  be chaotic
according   to   Devaney's   definition?  Firstly,   the   topological
transitivity property implies indecomposability.

\begin{definition} \label{def10}
A   dynamical   system   $\left(   \mathcal{X},  f\right)$   is   {\bf
indecomposable}  if it  is not  the  union of  two closed  sets $A,  B
\subset \mathcal{X}$ such that $f(A) \subset A, f(B) \subset B$.
\end{definition}

\noindent Hence, reducing the observed area
in order to  simplify its complexity, is impossible  if $\Gamma(f)$ is
strongly connected.  Moreover, under this  hypothesis the surveillance
scheme is strongly transitive:

\begin{definition} \label{def11}
A  dynamical system  $\left( \mathcal{X},  f\right)$ is  {\bf strongly
transitive} if $\forall x,y  \in \mathcal{X}$, $\forall r>0$, $\exists
z  \in  \mathcal{X}$,  $d(z,x)~\leq~r  \Rightarrow  \exists  n  \in
\mathds{N}^{\ast}$, $f^n(z)=y$.
\end{definition}
According to this  definition, for all pair of points  $x$, $y$ in the
phase space, a  point $z$ can be found in the  neighborhood of $x$ such
that one of its iterate $f^n(z)$  is $y$. Indeed, this result has been
stated   during   the  proof   of   the   transitivity  presented   in
\cite{guyeux09}. Among other things,  the strong transitivity leads to
the fact  that without the knowledge  of the initial  awaken nodes, all
scheduling are possible.  Additionally, no  nodes of the output space can
be  discarded  when  studying  the video-surveillance scheme: this 
 space  is  intrinsically
complicated and it cannot be decomposed or simplified.

Finally,  these WVSN possess  the  instability
property:
\begin{definition}
A dynamical  system $\left( \mathcal{X}, f\right)$ is  unstable if for
all  $x \in  \mathcal{X}$, the  orbit $\gamma_x:n  \in  \mathds{N} \longmapsto
f^n(x)$ is  unstable, that means: $\exists \varepsilon  > 0$, $\forall
\delta>0$, $\exists y \in \mathcal{X}$,  $\exists n \in \mathds{N}$, such that
$d(x,y)<\delta$    and   $d\left(\gamma_x(n),\gamma_y(n)\right)   \geq
\varepsilon$.
\end{definition}
This  property, which  is  implied by  sensitive  point dependence  on
initial conditions, leads to the fact that in all neighborhoods of any
point $x$ there  are points that can be apart  by $\varepsilon$ in the
future through iterations of the  WVSN. Thus, we can claim that
the behavior  of these  networks is unstable  when $\Gamma (f)$  is strongly
connected.

\subsubsection{Cryptanalysis in CKA Framework}

As stated in Section~\ref{subsec:Scheduling as Chaotic Iterations}, the 
proposed videosurveillance scheme can be rewritten as:

\begin{equation}
\left\{
\begin{array}{l}
X^{0}\in \mathcal{X} \\
X^{k+1}=G_{f_0}(X^{k}),
\end{array}
\right.
\end{equation}
where the phase space is $\mathcal{X}=\llbracket1;\mathsf{N}\times 2^\mathsf{N}\rrbracket^{\mathds{N}}\times \mathds{B}^{\mathsf{N}\times 2^\mathsf{N}}$, $X^0$ 
depends on a secret parameter $p=(p_1,\hdots,p_\mathsf{N})\in \left(\mathds{B}^\mathsf{N}\right)^\mathsf{N}$ 
whose binary digits are uniformy distributed, 
and $f_0$ stands for the vectorial negation on $\mathds{B}^{\mathsf{N}\times 2^\mathsf{N}}$.

We will now show that,

\begin{proposition}
The videosurveillance scheme proposed in this document is secure when facing a chosen key attack.
\end{proposition}

\begin{proof}
Let $N=\mathsf{N}\times 2^\mathsf{N}$.
We will prove by a mathematical induction that $\forall n \in \mathds{N}, X^n \sim
\mathbf{U}\left(\mathbb{B}^N\right)$. 
The base case is immediate, as the initial state of the
WVSN is initialized by $(g_1, ..., g_\mathsf{N}$, which
are producted by a cryptographically secure PRNG, so $X^0 \sim
\mathbf{U}\left(\mathbb{B}^N\right)$. 
Let us now suppose that the statement
 $X^n \sim
\mathbf{U}\left(\mathbb{B}^N\right)$ holds for some $n$. 
Let $e \in \mathbb{B}^N$ and $\mathbf{B}_k=(0,\cdots,0,1,0,\cdots,0) \in \mathbb{B}^N$ (the digit $1$ is in position $k$).
So $P\left(X^{n+1}=e\right)=\sum_{k=1}^N
P\left(X^n=e+\mathbf{B}_k,S^n=k\right).$
These two events are independent, thus:
$P\left(X^{n+1}=e\right)=\sum_{k=1}^N
P\left(X^n=e+\mathbf{B}_k\right) \times P\left(S^n=k\right)$.
According to the inductive hypothesis:
$P\left(X^{n+1}=e\right)=\frac{1}{2^N} \sum_{k=1}^N
 P\left(S^n=k\right)$.
The set of events $\left \{ S^n=k \right \}$ for $k \in \llbracket 1;N
\rrbracket$ is a partition of the universe of possible, so
$\sum_{k=1}^N P\left(S^n=k\right)=1$.

Finally, 
$P\left(X^{n+1}=e\right)=\frac{1}{2^N}$, which leads to $X^{n+1} \sim
\mathbf{U}\left(\mathbb{B}^N\right)$.
This result is true $\forall n \in \mathds{N}$, we thus have proven that, 
$$\forall p, Y_p=X^{N_0} \sim
\mathbf{U}\left(\mathbb{B}^N\right) $$

So the videosurveillance defined in this paper is secure in CKA.
\end{proof}

\section{Simulation Results}
\label{section: simulation results}

This section presents simulation results on comparing our chaotic approach to the 
standard C++ {\tt rand()}-based approach with random intrusions. We use the OMNET++ simulation environment and the next node selection will either use chaotic iterations or the C++ {\tt rand()} function ({\tt rand() \% $2^n$}) to produce a random number between 0 and $2^n$. For these set of simulations, $128$ sensor nodes (therefore $n=7$) are randomly deployed in a $75m * 75m$ area. Unless specified, sensors have an $36^o$ AoV and sensor node captures at the rate of 0.2fps. Each node starts with a battery level of 100 units and taking 1 picture consummes 1 unit of battery. When a node $V_i$ is selected to wake up, it will be awake for $T_i$ seconds. We set all $T_i=T=20s$. According to the behavior defined in section \ref{chaos-sch}, before going to sleep after an activity period of $e_i T$, $V_i$ will determine the next node to be waked up. It can potentially elect itself in which case $V_i$ stays active for at least another $T$ period. The elected node can be already active, in which case it simply increases its $e_i$ counter. We set about 50\% of the sensor nodes to be active initially (each sensor draws a random value between 0 and 1 and if the value is greater than 0.5, it will be active). This initial threshold is tunable but we did not try to vary this parameter in this paper. The results presented here have been averaged over 10 simulation runs with different initial seeds. Figure \ref{graph-activenode} shows the percentage of active nodes. Both the chaotic and the standard {\tt rand()} function have similar behavior: the percentage of active nodes progressively decreases due to battery shortage.

\begin{figure}  
\centering  
\includegraphics[width=\linewidth]{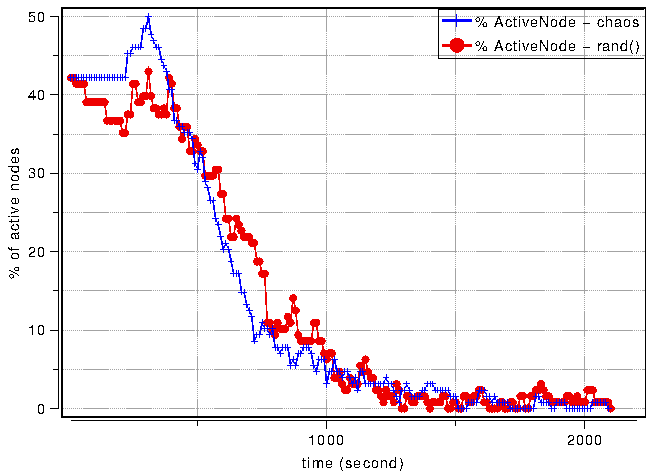}  
\caption{Percentage of active nodes.}   
\label{graph-activenode}  
\end{figure}

To compare both approaches in term of surveillance quality, we record to stealth time when intrusions are introduced in the area of interest. The stealth time is the time during which an intruder can travel in the field without being seen. The first intrusion starts at time 10s at a random position in the field. The scan line mobility model is then used with a constant velocity of 5m/s to make the intruder moving to the right part of the field. When the intruder is seen for the first time by a sensor, the stealth time is recorded and the mean stealth time computed. Then a new intrusion appears at another random position. This process is repeated until the simulation ends (i.e. no more sensor nodes with energy). Figure \ref{graph-stealthtime} shows the mean stealth time over the whole simulation duration. Figure \ref{graph-stealthtime-slidingwinavg} shows the same data but with a sliding window averaging filter of 20 values. As the nodes are uniformly distributed in the area of interest, there is a strong correlation between the percentage of active nodes and the stealth time as it can be expected. The result we want to highlight here is that our chaotic node selection approach has a similar level of performance in presence of random intrusions than standard {\tt rand()} function while providing a formal proof of non-prediction by malicious intruders.

\begin{figure}  
\centering  
\includegraphics[width=\linewidth]{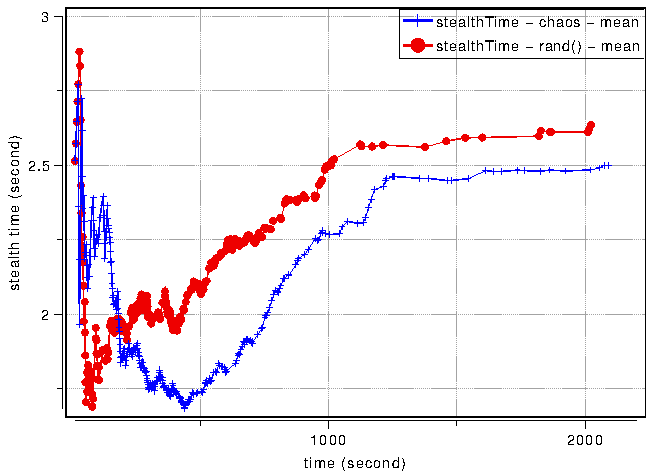}  
\caption{Stealth time.}   
\label{graph-stealthtime}  
\end{figure}

\begin{figure}
\centering  
\includegraphics[width=\linewidth]{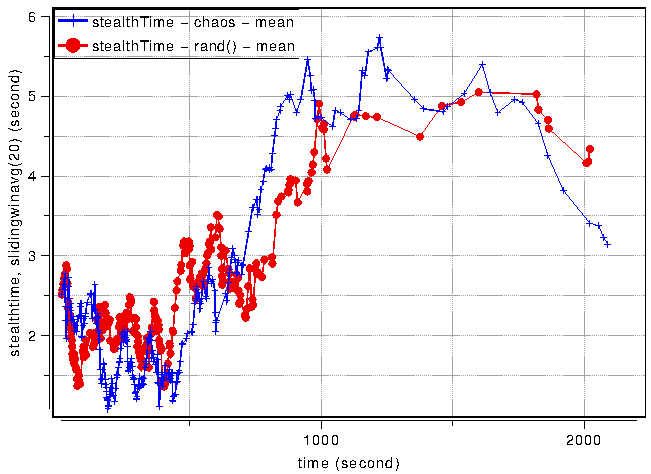}  
\caption{Stealth time.}   
\label{graph-stealthtime-slidingwinavg}  
\end{figure}

The last result we want to show is the energy consumption distribution. We recorded every 10s the energy level of each sensor node in the field and computed the mean and the standard deviation. Figure \ref{graph-stddev} shows the evolution of the standard deviation during the network lifetime. We can see that the chaotic node selection provides a slightly better distribution of activity than the 
standard {\tt rand()} function.

\begin{figure}
\centering  
\includegraphics[width=\linewidth]{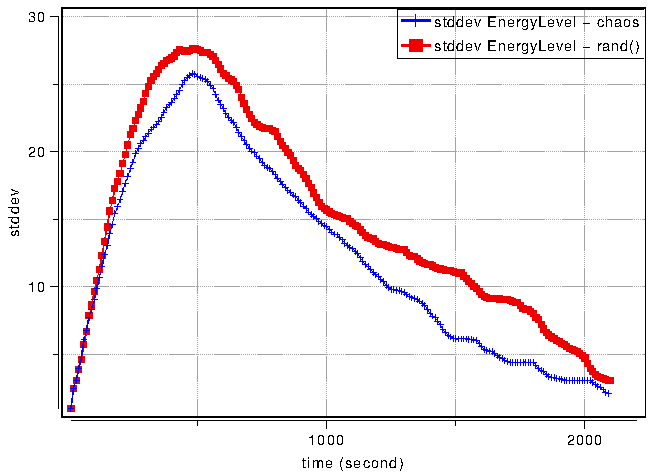}  
\caption{Evolution of the energy consumption's standard deviation.}   
\label{graph-stddev}  
\end{figure}

\section{Conclusions and Perspectives}

In this paper, a sleeping scheme for nodes has been proposed as an effective and secure solution to the joint scheduling problem in surveillance applications using WVSNs.
It has been evaluated through theoretical and practical aspects of performance and security. 
As opposed to existing works, this scheduling scheme is not based only on randomness, but on the mathematical theory of chaos too.
By doing so, we reinforce coverage and lifetime of the network, while obtaining a more secure scheme.
We have considered in this paper the case where the intruder is smart and active.
Furthermore, we have supposed that he can know the scheme and observe the behavior of the network. 
We have shown that, in addition of being able to preserve WVSN lifetime and to present comparable results against random attacks, our scheme is also able to withstand such malicious attacks due to its unpredictable behavior.

In future work, we intend to enlarge the security field in WVSN-based video surveillance, by making a classification of attacks that Oscar can achieve depending on the data he has access to.
Our desire is to distinguish between several levels of security into each category of malicious attacks, from the weakest one to the strongest one.
Additionally, we will study more precisely the topological properties of the scheduling scheme presented in this paper.

\bibliographystyle{plain}
\bibliography{mabase}%
\end{document}